\documentclass[12pt,a4paper]{amsart}

\usepackage[latin1]{inputenc}     % ��KK�SET K�YTT��N
\usepackage{amsmath,amssymb,amsthm}
\usepackage{latexsym}
\usepackage{amsfonts}
\usepackage{bbm,bm,dsfont} % for bold greek symbols, write \bm{\sigma} for example
\usepackage{color} 
\usepackage{graphicx}

\numberwithin{equation}{section} %mit� liene t�m� tarkoittaa?

\theoremstyle{definition}
\newtheorem{proposition}{Proposition}
\newtheorem{definition}{Definition}
\newtheorem{example}{Example}
\newtheorem{remark}{Remark}
\newtheorem{theorem}{Theorem}

\newcommand{\hA}{\mathcal{A}}

\newcommand{\hF}{\mathcal{F}}

\newcommand{\hU}{\mathcal{U}}

\newcommand{\R}{\mathbb R} %real
\newcommand{\C}{\mathbb C} %complex
\newcommand{\bS}{\mathbf{S}}
\newcommand{\bM}{\mathbf{M}}

\newcommand{\bL}{\mathbf{L}}

\newcommand{\hi}{\mathcal{H}} %Hilbert space
\newcommand{\ki}{\mathcal{K}} %other Hilbert space
\newcommand{\id}{\mathds1} %identity operator, MIK� OIS HYV�? PAKSU YKK�NEN? {\openone} ei toimi nyt!
\newcommand{\lh}{\mathcal{L(H)}} %bounded linear operators on H
\renewcommand{\th}{\mathcal{T(H)}} %trace class operators on H
\newcommand{\sh}{\mathcal{S(H)}} %states on H
\newcommand{\eh}{\mathcal{E(H)}} %effects
\newcommand{\tr}[1]{\mathrm{tr}\left[#1\right]} %trace
\def\<{\langle} %mainiot sulut! <
\def\>{\rangle} %mainiot sulut! >
\newcommand{\ket}[1]{|#1\rangle} %ket
\newcommand{\kb}[2]{|#1 \rangle\langle #2|} %ketbra
\newcommand{\Ao}{\mathsf{A}} %generic observable = \Eo
\newcommand{\Bo}{\mathsf{B}} 
\newcommand{\Co}{\mathsf{C}} 
\newcommand{\Eo}{\mathsf{E}} %generic observable = \sfe
\newcommand{\Fo}{\mathsf{F}} %generic observable = \sff
\newcommand{\Go}{\mathsf{G}} 
\newcommand{\Qo}{\mathsf{Q}} %generic observable = \sfq
\newcommand{\Po}{\mathsf{P}} %generic observable = \sfp
\newcommand{\Zo}{\mathsf{Z}}

\newcommand{\br}{\mathcal B(\mathbb R)} %reaalisuoran Borel sigma-algebra
\begin{document}

\title[Joint measurability in Lindbladian open quantum systems]{Joint measurability in Lindbladian open quantum systems} 

\author{Jukka Kiukas}
\address{Aberystwyth University, Aberystwyth SY23 3BZ, United Kingdom}
\email{jek20@aber.ac.uk}
\author{Pekka Lahti}
\address{University of Turku, Turku, Finland}
\email{pekka.lahti@utu.fi}
\author{Juha-Pekka Pellonp\"a\"a}
\email{juhpello@utu.fi}
\address{Department of Physics and Astronomy, University of Turku, FI-20014 Turku, Finland}

\begin{abstract} 
We study joint measurability of quantum observables in open systems governed by a master equation of Lindblad form. We briefly review the historical perspective of open systems and conceptual aspects of quantum measurements, focusing subsequently on describing emergent classicality under quantum decoherence. While decoherence in quantum states has been studied extensively in the past, the measurement side is much less understood --- here we present and extend some recent results on this topic.

%\noindent
%PACS numbers: 03.65.Ta, 03.67.--a
\end{abstract}

\maketitle

%%%%%%%%%%%%%%%%%%%%%%%%%%%%%%%%%%%%%%%%%%%%%%%%%%%%%%%%%%%

\section{Introduction}

Incompatibility, the lack of joint measurability, is a  key feature of quantum mechanics.
It appeared as the noncommutativity of  (the mathematical expressions of) the basic physical quantities of a quantum system 
already in the very first papers of Heisenberg \cite{Heisenberg1925} and Schr\"odinger \cite{Schrodinger1926} on quantum mechanics,
with an intuitive understanding that such quantities cannot be measured  jointly unless a sufficient amount of mutual inaccuracies are allowed in their measurements, inaccuracies expressed in terms of the  uncertainty relations \cite{Heisenberg1927}.
The first part of this, the noncommutativity and its relation to the lack of joint measurability, was clarified and made rigorous in the work of von Neumann \cite{vN1932}.
The clarification of the second part of this, the possibility for an approximate joint measurability, had to wait for an extension of the mathematical representation of the notion of a physical quantity, observable, as a selfadjoint operator, real spectral measure, to a normalized positive operator measure, semispectral measure, 
as well as for an appropriate notion of an approximate measurement in quantum mechanics. The identification of quantum observables as normalized positive operator measures
 emerged naturally  in an operationally motivated axiomatic reconstructions/reformulations of quantum mechanics initiated notably by Ludwig \cite{Ludwigs}, Davies and Lewis \cite{DaviesLewis1970}, as well as Ingarden \cite{Ingarden}.
This extension made it also clear that a quantum mechanically meaningful notion of an approximate measurement of an observable is to be based on a comparison between the statistics of outcomes of two measurements (observables), the actually performed one and the targeted one, 
an idea clearly expressed by Ludwig \cite[pp.\ 197-8]{Ludwig1983} and worked out, 
for instance, \cite{Werner2004,BLW2014}.

Interaction of a system with another system is the source of information  but also influence (designed or not) on the system. In addition to abrupt measurements, evolution of an open quantum system, the reduced dynamics, is an instance  of such a possible influence. When an observable is measured on a dynamical quantum system, the effect of the evolution taking place between a fixed initial time and the time of the measurement can be incorporated into the observable in the Heisenberg picture, allowing for initially incompatible observables to become jointly measurable at a later time.

Apart from some simple cases, the reduced dynamics obtained from the unitary evolution of the total system, the system $\bS$ and its environment, though uniquely defined, is, however, typically quite involved and does not  easily lend itself to any practical use, in particular, ``as [it] depend[s] on the chosen initial state of [its environment]  at a particular instant in the past.  Only with some additional assumptions will there be an autonomous time-homogeneous evolution in $\bS$. One necessary condition is clearly that the state of [the environment] does not change significantly as a result of the interaction with  $\bS$.''  \cite{Lindblad1983}. 
An important instance of such subsystem dynamics is given by a uniformly continuous semigroup of dynamical maps 
in which case the state of $\bS$  evolves according to the  equation 
derived by G\"oran Lindblad in 1976 
\cite{Lindblad1976} and since then generally known as the Lindblad or the  Gorini, Kossakowski, Lindblad, Sudarshan  (GKLS) equation, due to the parallel independent work of Gorini, Kossakowski,  and Sudarshan \cite{GKS76}.
For a study of the chain of events, intuitions and ideas that led to the formulation of these equations, we refer to \cite{GKLS}.

The semigroup dynamics warrants that if some initially incompatible observables turn jointly measurable at some time, then they will remain so for all subsequent times. Hence one obtains a unique (possibly infinite) critical time after which incompatibility is irreversibly lost, reflecting the Markovian character of the evolution. Given the fundamental role of incompatibility within quantum theory, this feature can be seen as one aspect of ``emergent classicality'' in open quantum systems, studied much less than, say, loss of entanglement in quantum states. Of course, determining the critical times exactly is typically a hard problem; however, as we will demonstrate below, one may proceed by deriving analytical bounds.

The structure of the paper is the following. We begin by reviewing the relevant mathematical and conceptual framework of quantum measurement theory and open systems and discussing general aspects of joint measurability under quantum dynamics. We then focus on the situation where the dynamics exhibits emergence of classicality, in the sense of driving the system into a commutative decoherence-free subalgebra in the long-time limit. In particular, we demonstrate that under generic assumptions, any pair of observables will become jointly measurable at some finite critical time before the system becomes commutative, and illustrate this in specific examples where good analytical bounds on the critical time can be obtained.

\section{Preliminaries on quantum measurements and channels}

We introduce here the necessary notations and concepts related to quantum measurements, joint measurability, and quantum channels.

Let  $\hi,$ $\lh,$ $\th$ be the (complex separable) Hilbert space  and  the Banach spaces of bounded and trace class operators on which the description of a physical systems $\mathbf S$ is based, with the notions of states $\rho\in\sh$, positive trace one operators, and observables $\Eo:\hA\to\lh$, normalized positive operator measures with value (measurable) spaces $(\Omega,\hA)$, describing the measurement outcome statistics  $\Eo_\rho$, with $\Eo_\rho(X)=\tr{\rho\Eo(X)},$ $X\in\hA$,
for the observable $\Eo$ in the state $\rho$. We use $\id_\hi$ or briefly $\id$ for the identity operator of $\hi$. 

The notion of compatibility, the joint measurability, of any two observables $\Eo_1$ and $\Eo_2$, with the respective value spaces $(\Omega_i,\hA_i)$, $i=1,2$, can be defined in several equivalent ways, see, e.g.\ \cite{QM}; for instance, $\Eo_1,$ $\Eo_2$ are jointly measurable exactly when they have a joint observable, that is, there is an observable $\Go:\hA_1\otimes\hA_2\to\lh$ such that $\Go(X\times\Omega_2)=\Eo_1(X)$ and $\Go(\Omega_1\times Y)=\Eo_2(Y)$ for all $X\in\hA_1, Y\in\hA_2$. Observables $\Eo_1$ and $\Eo_2$ are called incompatible if they are not jointly measurable. In the case of projection valued observables $\Eo_1$ and $\Eo_2$
with standard Borel value spaces,  $\Eo_1$ and $\Eo_2$ are compatible
if and only if they commute, that is,
$\Eo_1(X)\Eo_2(Y)=\Eo_2(Y)\Eo_1(X)$ for all $X,$ $Y$, in which case $(X,Y)\mapsto\Eo_1(X)\Eo_2(Y)$ defines their unique joint observable.
In general, the mutual commutativity of $\Eo_1$ and $\Eo_2$ is not needed for their compatibility. 

In many concrete applications, the value spaces of the observables are the real Borel spaces $(\R,\br)$, their subspaces or their Cartesian products. For real projection valued observables we use the symbols $\Ao,\,\Bo,\,\Co$, and also identify them with the uniquely associated selfadjoint operators $A,\,B,\,C$, with, for instance, $A=\int a\Ao(da)$. Occasionally, we refer to such observables as sharp, and others as unsharp.

To describe the changes experienced by a system, we can use alternatively both the Schr\"odinger and the Heisenberg pictures of quantum channels, either as a change of states or as a change of observables. In the Heisenberg picture, a quantum channel is a normal unital completely positive linear map $\Lambda:\lh\to \lh$. Due to normality, it has the predual map $V:\th\to \th$ determined by ${\rm tr}[V(\rho)A]= {\rm tr}[\rho\Lambda(A)]$ for all $\rho\in\th$ and $A\in \lh$. For any Heisenberg channel $\Lambda$, and any observable (measurement) $\Eo$, we can form a new observable (measurement) $\Lambda(\Eo):=\Lambda\circ \Eo$, whose operational interpretation is given by the duality: if we measure $\Eo$ on the state $V(\rho)$ obtained from a given initial state $\rho$ by the application of the Schr\"odinger channel $V$, then the outcome distribution $\Eo_{V(\rho)}$ coincides with $\Lambda(\Eo)_{\rho}$, that is, the distribution of the observable $\Lambda(\Eo)$ in the initial state $\rho$.

\section{Open quantum systems}

In this section we review some basic aspects of open quantum systems, to the extent relevant to the joint measurability problems introduced in the subsequent settings.

\subsection{General remarks on subsystem dynamics}\label{sec:subgen}

In its fundamental formulation, open quantum system (see e.g.\ \cite{davies,alicki,breuer}) consists of an object system $\bS$ with Hilbert space $\mathcal H$ coupled to an environment $\bM$ with Hilbert space $\mathcal K$, such that the total system $\bS+\bM$ is assumed closed with evolution governed by some (strongly) continuous unitary representation $\R\ni t\mapsto U_t\in\hU(\hi\otimes\ki)$, with $\hU(\hi\otimes\ki)$ being the group of unitary operators on $\hi\otimes\ki$.
However, as well-known,  this approach is not very useful  in studying the evolution of the object system. 
A first step in that direction is to assume that the state of the environment is known at time $ t=0$, say, and that $\bS$
 can then independently be prepared in any  state $\rho\in\sh$.\footnote{The apparent inconsistency of this assumption with the assumption concerning the evolution of $\bS+\bM$ is easily settled by extending  (assuming invertibility) the dynamical semigroup of  the total system to a dynamical group and then applying the Wigner theorem (and the fact that all multipliers of $(\R,+)$ are exact) to reach back the initial assumption. 
For a full proof, see, e.g., \cite{simon,Gianni}.} 
If $\sigma$ is the state of $\bM$ at $t=0$, then, with the given assumption, 
the state of $\bS$ at any time $t$ is simply
\begin{equation}\label{redS}
V_t(\rho)={\rm tr_\ki}[U_t(\rho\otimes\sigma)U_t^*]
\end{equation}
where ${\rm tr_\ki}$ is the partial trace over $\ki$. 
The thus defined map $V_t$ is a quantum channel in its Schr\"odinger representation. If it would be invertible, then, by the Wigner theorem,  it could be implemented by a unitary or antiunitary operator. In general, this is not the case, that is, as a rule, $V_{-t}\ne V_t^{-1}$, in which case the  composition rule $V_{-t}\circ V_t=V_{-t+t}$ cannot hold. We recall, in addition,  that the  map $\R^+\ni t\mapsto V_t\in\th$ is strongly continuous.

Equivalently, in the Heisenberg  picture one has  the (strongly continuous) map  $t\mapsto\Lambda_t$, consisting of the 
quantum channels, given by
\begin{equation}\label{red}
\Lambda_t(A) = \mathcal E_\sigma (U_t^*(A\otimes \id) U_t)
\end{equation}
where $\mathcal E_\sigma:\mathcal L(\mathcal H\otimes \mathcal K)\to \mathcal L(\mathcal H)$ is the conditional expectation on the state $\sigma$, defined by the formula ${\rm tr}[(\rho\otimes \sigma)Z]={\rm tr}[\rho \mathcal E_\sigma(Z)]$. Again, in general, neither  the maps $\Lambda_t$ are  invertible nor the composition rule $\Lambda_{t'}\circ\Lambda_t=\Lambda_{t+t'}$ holds. 

Vast amount of  work has gone to justify  and specify further the possible structure of the subsystem dynamics arising from \eqref{redS}, or, equivalently,  \eqref{red}, and in that the work of G\"oran  Lindblad is of lasting importance. Following Lindblad's  pioneering paper  \cite{Lindblad1976} we shall briefly  summarize the prototype example  of a uniformly continuous dynamical semigroup.

\subsection{Uniformly continuous semigroups}
The (strongly continuous) dynamical map $t\mapsto V_t,$ $t\geq 0$, resp.\  $t\mapsto\Lambda_t,$ $t\geq 0$, obtained through \eqref{redS}, resp.\ \eqref{red}, is  very rarely a semigroup, but in several cases it can be approximated by such with reasonable accuracy \cite{davies,alicki,breuer,spohn}. For this reason, and due to its powerful mathematical implications, the semigroup dynamics is often used ``phenomenologically'', without considering the precise conditions under which the approximation might be valid. 

Assume now that $\{\Lambda_t\,|\,  t\geq 0 \}$  is a  semigroup, that is, $\Lambda_0={\rm Id}$, and $\Lambda_{t+t'}=\Lambda_{t'}\circ\Lambda_t$ for all $t,\,t'\ge 0$. 
As well known, it has  a densely defined generator and that the generator is a bounded linear map $\bL:\lh\to\lh$ exactly when the map $t\mapsto\Lambda_t$ is uniformly continuous,  in which case
$\bL=\lim_{h\to 0+}\frac 1h(\Lambda_h-I)$ and $\Lambda_t=e^{t\bL}$ for all $t\geq 0$ \cite{DS1958}. By the complete positivity of the maps $\Lambda_t$, the above equivalence gets a more specific expression: the semigroup is uniformly continuous exactly when its generator $\bL$
is of the form
\begin{equation}\label{GSKL}
\bL(A)=i[H,A]+\sum_l\left(L_l^*AL_l -\frac{1}{2}\{L_l^*L_l,A\}\right)
\end{equation}
for some selfadjoint  $H\in\lh$ and  countable $(L_l)\subset\lh$ with $\sum_lL^*_lL_l\in\lh$ (weakly, if the sum infinite).  Here $[\;,\,]$ and $\{\;,\,\}$ are the usual notions for commutator and anticommutator, respectively.
The operators $L_l$ are often called Lindblad operators.

In the Schr\"odinger picture, the generator of the dual semigroup $\{V_t\,|\,t\geq 0\}$ takes the form
\begin{equation}\label{GSKLstates}
\tilde\bL(\rho)=-i[H,\rho]+\sum_l\left(L_l\rho L_l^*-\frac 12\{L_l^*L_l,\rho\}  \right)
\end{equation}
The form of the generator \eqref{GSKL}, resp.\  \eqref{GSKLstates}, is known as the GSKL form  and the equations
\begin{eqnarray}
\frac d{dt}\Lambda_t(A)&=&\bL(\Lambda_t(A))\\
\frac d{dt}V_t(\rho)&=&\tilde\bL(V_t(\rho))
\end{eqnarray}
are known as the {\em Lindblad equations} with the initial conditions $\Lambda_0(A)=A$ and $V_0(\rho)=\rho$.
For the progress in the  study of the notion of unbounded GKLS generators we refer to
\cite{Davies1979,SHW}.

\subsection{The case of general decoherence}
We will focus mainly on a setting where the effect of \emph{decoherence} can be isolated from the overall evolution --- this will make it possible to discuss joint measurability in some detail later. We first recall here that the term ``decoherence'' is often used loosely, to mean the overall effect associated with the system being open. When referring specifically to the damping of the off-diagonal elements of operators in some fixed basis, \emph{pure decoherence} can be used \cite{breuer}. 
The starting point is the \emph{decoherence-free subalgebra} 
$$
\mathcal C=\{A\in \lh\mid \Lambda_t(A^*A)=\Lambda_t(A)^*\Lambda_t(A),\, \Lambda_t(AA^*)=\Lambda_t(A)\Lambda_t(A)^* \text{ for all }t\geq 0\}
$$
(studied since the 1970s, see e.g.\ \cite{Evans1977}), where $\Lambda_t$ behaves as if the system was closed, and the usual aim is to characterise situations where the evolution vanishes asymptotically at $t\rightarrow \infty$ in a suitable complemented subspace of $\mathcal C$, so that the effect of ``openness'' of the system is separated from the automorphic part. This setting has been studied by several authors over the past decades in various levels of generality \cite{Blanchard2003A, Blanchard2003B, Carbone2013, Carbone2015, Deschamps2016}; we give the precise statement here in the comparatively simple case where $\mathcal C$ is commutative, which is relevant for our aim. We further exclude the case where one cannot ``project'' observables onto the classical system $\mathcal C$, by assuming that $\mathcal C$ is atomic\footnote{That is, any projection of $\mathcal C$ has a minimal subprojection (in $\mathcal C$).} and the semigroup has a faithful invariant state.\footnote{A state $\sigma\in\sh$ is faithful if $\ker\sigma=\{0\}$, and  invariant if $V_t(\sigma)=\sigma$ for all $t\ge 0$.}
 Finally, in order to make an explicit connection to the Lindblad form, we take the semigroup to be uniformly continuous. Then the following result holds:

\begin{theorem}[\cite{Carbone2015,Deschamps2016}]\label{structure} Suppose that $\Lambda_t$ is uniformly continuous, has a faithful invariant state $\sigma\in\sh$, and the algebra $\mathcal C$ is atomic and Abelian with a resolution of identity 
 $(P_n)_{n\in I}$ consisting of minimal projections of $\mathcal C$. Denoting $\mathcal H_n = {\rm ran}\, P_n$, so that $\mathcal H = \oplus_{n\in I} \mathcal H_n$ and $\mathcal C= \oplus_{n\in I} \mathbb C\id_{\hi_n}$, the following properties hold:
\begin{itemize}
\item[(a)] The Hamiltonian and the Lindblad operators decompose along the direct sum: $L_l =\oplus_{n\in I} L_l^{(n)}$ and $H = \oplus_{n\in I} K^{(n)}$, where $L_l^{(n)},K^{(n)}\in \mathcal L(\hi_n)$.
\item[(b)] The restriction of $\Lambda_t$ to each invariant subalgebra $P_n\lh P_n$ has a unique faithful invariant state $\sigma_n$, and 
$\sigma = \oplus_{n\in I} {\rm tr}[\sigma P_n] \sigma_n$.
Furthermore, the map $\Gamma:\lh \to \mathcal C$ given by
$\Gamma(A) = \oplus_{n\in I} {\rm tr}[\sigma_n P_nAP_n]P_n$
is a normal conditional expectation\footnote{We recall that a \emph{conditional expectation} $\Gamma:\,\lh\to\mathcal A$, where $\mathcal A$ is a von Neumann subalgebra of $\lh$, is a linear positive map such that $\Gamma(A)=A$ and $\Gamma(AB\tilde A)=A\Gamma(B)\tilde A$ for all $A,\,\tilde A\in\mathcal A$ and $B\in\lh$. Moreover, it follows that $\Gamma$ is a 
unital idempotent.
Existence of \emph{normal} conditional expectations is a rather tricky matter, generally relying on the Tomita-Takesaki theory \cite{Takesaki1971}. In particular, normal conditional expectations do not generally exist for non-atomic Abelian subalgebras of $\lh$.
} onto $\mathcal C$ such that $\Lambda_t\circ\Gamma=\Gamma\circ\Lambda_t$ for all $t\ge 0$ and ${\rm tr}[\sigma \Gamma(A)]={\rm tr}[\sigma A]$ for all $A\in \lh$.
\item[(c)] $\ker \Gamma$ is a weak-* closed $*$--invariant and $\Lambda_t$--invariant subspace of $\lh$ such that $\lh = \mathcal C \oplus \ker \Gamma$, and ${\rm w}^*{-}\lim_{t\rightarrow\infty} \Lambda_t(A)=0$ for all $A\in \ker \Gamma$.
\end{itemize}
\end{theorem}

Here part (c) expresses the ``emergence of classicality'' property in this setting -- for each observable $\Eo$, only the observable $\Gamma(\Eo)$ survives in the limit, and these are all mutually commutative, belonging to the algebra $\mathcal C$. Note that here the normality of the conditional expectation is crucial since it ensures (together with the positivity of $\Gamma$) that $\Gamma(\Eo)$ is indeed an observable.

\subsection{The case of pure decoherence}\label{pure1}

It is useful to write down explicitly the instance of the above structure in which the decoherence-free subalgebra $\mathcal C$ is maximal Abelian, that is, equal to the algebra of diagonal operators $\sum_n a_n |n\rangle\langle n|$ in a fixed basis $\{|n\rangle\mid n\in I\}$. In this case each $\mathcal H_n=\C|n\>$ is one-dimensional (and we restrict to a finite-dimensional total Hilbert space for technical simplicity). From Theorem \ref{structure} part (a) we then find that each Lindblad operator $L_l$ must be diagonal, that is,
$$
L_l = \sum_{n} u^{(l)}_n |n\rangle\langle n|,\quad H = \sum_n h_n |n\rangle\langle n|
$$
where $u^{(l)}_n\in \mathbb C$, $h_n\in \mathbb R$ are constants. It follows that
$\bL(A) = D \star A$, where $\star$ stands for the Schur (elementwise) product of matrices in the basis $\{|n\rangle\mid n\in I\}$, with the ``generator'' $D=(d_{nm})$ given by
$$
d_{nm} = i(h_n-h_m) -\frac 12\sum_{l} \Big(|u^{(l)}_n|^2+ |u^{(l)}_m|^2-2\overline{u^{(l)}_n}u^{(l)}_m\Big).
$$
It follows then immediately that the generating semigroup also has the Schur product form $\Lambda_t(A) = C(t) \star A$, where
$$
c_{nm}(t) = e^{t d_{nm}} = \langle \varphi_n(t)|\varphi_m(t)\rangle
$$
where the last form shows explicitly that the multiplier matrix $C(t)$ is positive for all $t\geq 0$, being expressed in terms of scalar products of the vectors
$$
\varphi_n(t) = e^{-ith_n} |\sqrt{t} u_n^{(1)}\rangle\otimes \cdots \otimes |\sqrt{t} u_n^{(K)}\rangle\in L^2(\mathbb R)^{\otimes K}
$$
in the $K$-fold tensor product of standard coherent states $|z\>$ of the harmonic oscillator for which 
$\<z|z'\>=e^{-(|z|^2|+|z'|^2-2\overline z z')/2}$ holds. It seems  that this dilation does not in general have a representation in terms of a unitary evolution in the combined system consisting of the original system and the  dilation space $L^2(\mathbb R)^{\otimes K}$.

Finally, we need to characterise the crucial assumption that the decoherence-free subalgebra $\mathcal C$ coincides with the diagonal algebra:
$$\mathcal C = {\rm span}\{|n\rangle\langle n| \mid n\in I\}.$$

Note first that each $|n\rangle\langle n|$ is clearly unchanged by all $\Lambda_t$, so that $|n\rangle\langle n|\in \mathcal C$ for all $n$ in any case. If $A\in \lh$ satisfies $\Lambda_t(A^*A)=\Lambda_t(A)^*\Lambda_t(A)$ then
$$
\langle n|A^*A|n\rangle = \sum_{k} |c_{nk}(t)|^2 |\langle n|A|k\rangle|^2,
$$
and the second condition in the definition of $\mathcal C$ gives a similar expression. 
It follows immediately that $\mathcal C$ contains a non-diagonal operator iff $|c_{nk}(t)|^2 =1$ for some pair $n\neq k$. Therefore, the relevant condition for decoherence is
$|c_{nk}(t)|^2 = e^{2t {\rm Re}\,d_{nk} }<1$ for all $n\neq k$, which is equivalent to
$${\rm Re}\,d_{nm}=-\frac12\sum_l |u^{(l)}_n-u^{(l)}_m|^2<0 \quad \text{for all }n\neq m.$$
We now observe directly that this is exactly the required decoherence condition
$$\lim_{t\rightarrow\infty}\Lambda_t = \Gamma,$$
implied by Theorem \ref{structure}. The central projections of $\mathcal C$ are then $P_n=|n\rangle\langle n|$, the faithful invariant states are of the form
$$
\sigma = \sum_n p_n |n\rangle\langle n|
$$
where $p_n>0$ for all $n$, and the conditional expectation reads
$$
\Gamma(A) = \sum_n \langle n|A|n\rangle |n\rangle\langle n|.
$$

\section{Joint measurability under subsystem dynamics}

We now proceed to introduce the measurement aspect to open system dynamics, in the subsystem setting described in Section \ref{sec:subgen}. The operational context is the following: suppose the system is initialised at time $t=0$ in a state $\rho$, and we measure an observable $\mathsf E$ at time $t$; then the probability measure for the outcomes is $X\mapsto {\rm tr}[V_t(\rho)\Eo(X)]={\rm tr}[\Lambda_t(\mathsf E)\rho]$, which is equivalent to measuring the Heisenberg-evolved observable $\Lambda_t(\mathsf E)$ on the initial state $\rho$. Joint measurability of observables in this dynamical setting corresponds to the emergence of a hidden variable model from which the joint measurement outcomes can be deduced; this means, for instance, that Bell inequality violations or quantum steering (involving some coupled system) is no longer possible using these measurements \cite{Kiukas2017}. 

There are two conceptually different scenarios: we can either use the other system to implement the measurement on the object system, or regard the other system as an ambient environment which merely adds noise to the system. We consider briefly these scenarios in the following two subsections. In both cases, the starting point for joint measurability considerations is the following simple observation that measurement incompatibility cannot be ``created'' by the action of a quantum channel. Due to the importance of this fact in the subsequent development we record it in the following proposition. 

\begin{proposition}\label{JMLambdaprop} Let $\Lambda$ be a quantum channel and $\mathsf E,\,\mathsf F$ jointly measurable observables. Then $\Lambda(\mathsf E)$ and $\Lambda(\mathsf F)$ are jointly measurable.
\end{proposition}
\begin{proof}
Suppose that a pair of observables $\mathsf E,\,\mathsf F$ are jointly measurable with a joint observable $\mathsf G$, so that $\mathsf E(X) = \mathsf G(X\times \Omega)$ and $\mathsf F(Y)=\mathsf G(\Omega\times Y)$. Then clearly $\Lambda(\mathsf F)(X)= \Lambda(\mathsf G(X\times \Omega))$ and $\Lambda(\mathsf F)(X)= \Lambda(\mathsf G(\Omega\times Y))$, so $\Lambda(\mathsf G)$ is a joint observable for $\Lambda(\mathsf E)$ and $\Lambda(\mathsf F)$. 
\end{proof}

\subsection{Measurement interaction}\label{MeasInt}

Standard way of implementing a general quantum measurement is to couple the object system to an ancillary pointer system through a unitary interaction and measuring there; the object system is then manifestly open, and one can describe the actually measured observable in terms of the reduced subsystem dynamics. We can then ask when two observables implemented this way are jointly measurable.

To formalise and to illustrate the setting, let $\bS$ interact with another system $\bM$ through a unitary coupling of the form $U=e^{i A\otimes B}$, where $A$ and $B$
are the selfadjoint operators corresponding to the sharp observables $\Ao$ and $\Bo$ of $\bS$ and $\bM$, respectively, and assume 
that the two systems are initially (prior to the interaction) independent of each other,
 the total system being in a product state $\rho\otimes\sigma$ (which could be taken to be pure).

If $\bM$ with $U$ is to model a measurement of $\bS$ one needs, in addition,  a pointer observable $\Zo$, with a value space $(\Xi,\hF)$,   to determine  the observable  reproduced by its outcome statistics, that is, to find out the observable $\Eo$ of $\bS$ for which  
$\Eo_\rho(X)=\Zo_{\tilde\sigma}(X)$, with $\tilde\sigma={\rm tr}_\hi[U(\rho\otimes\sigma)U^*]$, for all $X$ and $\rho$.  
A simple computation gives\footnote{ We recall: $e^{iA\otimes B}=\int\int e^{iab}\Ao(da)\otimes\Bo(db)
=\int e^{ibA}\otimes\Bo(db) =\int \Ao(da)\otimes e^{iaB}$. }
$$
\Eo(X)=\int \Zo_{e^{iaB}\sigma e^{-iaB}}(X)\,\Ao(da),\qquad X\in\Xi,
$$
which shows that the actually measured observable $\Eo$ is a smeared version of the observable $\Ao$, smeared with the Markov kernel $(X,a)\mapsto p(X,a)=\tr{e^{iaB}\sigma e^{-iaB}\Zo(X)}$. Clearly, if $\Zo$ can be chosen to be covariant under $a\mapsto e^{-iaB}$, that is, with $\hF=\br$,  
 $e^{-iaB}\Zo(X)e^{iaB}=\Zo(X-a)$, 
then $\Eo$ is just the convolution of $\Ao$ with the initial pointer distribution $\Zo_\sigma$; $\Eo=\Zo_\sigma*\Ao$.
This model was used already by von Neumann \cite[pp.\,236-7]{vN1932} to demonstrate the possibility of measuring (within his  newly developed theory of a measurement process)  the position $Q$ of a system $\bS$
choosing
$A=Q$ (for $\bS$) and $B=P$ (momentum), $Z=Q$ (for  $\bM$) and choosing $\sigma$ such that the initial pointer distribution $\Zo_\sigma=\Qo_\sigma$ approaches a point measure, in which case the measured $\bS$--observable  $\Zo_\sigma* \Qo$ approaches its sharp position observable $\Qo$.

On the other hand, if $\bM$ together with $U$  is to model just an environment of the system then one aims to determine 
its influence on $\bS$ without any use of a pointer.
This  manifests itself directly  in the change of the measurement statistics of any of  its observables:
for any $\rho$, $\Eo$, and $X\in\hA$,
$$
\tr{\rho\Eo(X)}\quad\Longrightarrow \quad\tr{V(\rho)\Eo(X)}=\tr{\rho\Lambda(\Eo(X))},
$$
where, with  \eqref{redS} and \eqref{red},
\begin{eqnarray*}
 V(\rho)&=& \int_\R e^{ibA}\rho e^{-ibA}\, \Bo_\sigma(db)\\
 \Lambda(\Eo(X))&=&
\int_\R  e^{-ibA}\Eo(X)e^{ibA}\, \Bo_\sigma(db).
\end{eqnarray*}
Clearly, if $\Eo$ and $\Ao$ commute with each other, then  $\Lambda(\Eo)=\Eo$.
Moreover, if $\Eo$ and $\Fo$ are jointly measurable, then by Proposition \ref{JMLambdaprop},  also $\Lambda(\Eo)$ and $\Lambda(\Fo)$ are jointly measurable. However, it may well happen that the blurred observables $\Lambda(\Eo)$ and $\Lambda(\Fo)$ are compatible even though $\Eo$ and $\Fo$ were not (see example, below).
We note further that
if $\Eo$ happens to be   covariant under $b\mapsto e^{-ibA}$,  then %hat is, $e^{-ib A}\Eo(X)e^{ibA}=\Eo(X-b)$, 
$\Lambda(\Eo)$ is a smearing (convolution) of $\Eo$ with the probability measure $\Bo_\sigma$, $\Lambda(\Eo)=\Bo_\sigma*\Eo$.

With the choice $A=Q$ (for $\bS$) and $B=P$ (for $\bM$) one sees, for instance, that the momentum $P$ of system $\bS$ turns unsharp: 
$\Po\mapsto\Lambda(\Po)= \Po_\sigma*\Po$, convolution of the momentum  of $\bS$ with the momentum distribution of $\bM$.
On the other hand, the `conjugate' choice $A=P$ (for $\bS$) and $B=Q$ (for $\bM$) would  mold  the position  of $\bS$ to an unsharp one, $\Lambda(\Qo)=\Qo_\sigma *\Qo$, blurred with the position distribution $\Qo_\sigma$  of $\bM$. As well known, no such pair $(\Qo,\Lambda(\Po))$ or $(\Lambda(\Qo),\Po)$ is compatible.

Consider next the case where $\bM$ consists of two parts, $\bM=\bM_1+ \bM_2$, with the Hilbert space $\ki_1\otimes\ki_2$ and
with the initial state $\sigma=\sigma_1\otimes\sigma_2$, and let $\bS$ interact with $\bM$ via the Arthurs-Kelly coupling, conveniently written as
$$
U=U'U_2U_1=e^{\frac i2 \id\otimes P_1\otimes Q_2} e^{iP\otimes\id_1\otimes Q_2}e^{-iQ\otimes P_1\otimes\id_2}.
$$
In this interaction, the momentum $P$ and the position $Q$ of $\bS$ get changed to
$\Po\mapsto \Lambda(\Po)=\Po_{1,\sigma_1}*\Po$ (due to $U_1$)
 and $\Qo\mapsto\Lambda(\Qo)= \Qo_{2,\sigma_2}*\Qo$ (due to $U_2$), respectively (with no influence from the $\bM_1-\bM_2$ coupling).

Such a pair of unsharp position and momentum is not necessarily jointly measurable. 
In fact,
they
are known to be compatible  exactly when the convolving propability measures 
$\Po_{1,\sigma_1}$ and $\Qo_{2,\sigma_2}$ are Fourier related, that is, $\Po_{1,\sigma_1}=\Po_T$ and $\Qo_{2,\sigma_2}=\Qo_T$ for some $T\in\sh$ \cite{CHT2005}.  With an appropriate choice of the initial states $\sigma_1$ and $\sigma_2$ this condition can be met.

 It is another well known fact that if one uses $\bM$ with the Arthurs-Kelly coupling as a measuring device, with the pointer observable $Q_1\otimes P_2$, then the measured observable of $\bS$ is a biobservable with the two marginal observables $\Qo_{1,\sigma_1}*\Qo$ and $\Po_{1,\sigma_1}*\Po_{2,\sigma_2}*\Po$.
With this choice of the pointer, the $\bM_1-\bM_2$ coupling now adjusts the convolving probability measures  $\Qo_{1,\sigma_1}$ and $\Po_{1,\sigma_1}*\Po_{2,\sigma_2}$ of $\Qo$ and $\Po$ 
such that the resulting unsharp position and momentum have a joint observable (given by the measurement)  for any choice of $\sigma_1$ and $\sigma_2$.
Eight-port homodyne detection gives an optical implementation of such an  approximate joint measurement of position and momentum. The study of this joint measurement model goes back to Arthurs and Kelly \cite{ArthursKelly1965}, for further elaboration and  details, see, e.g.\ \cite [Sec.\ 19.1, 19.4]{QM}.

\begin{remark}
In his unfinished manuscript  A Neo-Copenhagen Quantum Mechanics,
published in this special volume of OSID, 
G\"oran Lindblad  writes: ``[the] standard picture [of a measurement] must be rejected'', since 
``[i]t is really  a process of copying a part of the information
from the state of [$\bS$] to the state of [$\bM$], both equally quantum'' and does not
``describe the outcome of the
measurement as an event with classical, factlike properties like stability and objectivity'',
 Lindblad \cite{Lindblad2023}.  
A  part from its limitations, we wish to emphasize that the probability reproducibility condition
$\Eo_\rho=\Zo_{\tilde\sigma},$ $\rho\in\sh$, which we use to define the actually measured observable
is highly useful for an operational definition of the object observables, and this, in particular, since each object observable admits such a representation, see, e.g., \cite[Theorem 10.1]{QM}.
Clearly, its does not describe the occurence of an actual result, a reason why such measurement schemes are, as Lindblad [{\em ibid.}] notes,  occasionally called premeasurements, see, e.g. \cite{BCL,BLM}.
\end{remark}

\subsection{Environmental interaction} Here we ignore the implementation of the measurements, focusing only on the time evolution of the measurements in the object system, the aim being to investigate how joint measurability changes under the dynamics. Obviously, this depends heavily in the specifics of the dynamics --- the purpose of the present section is to illustrate this with a simple non-Markovian example (adapted from \cite{EnricoGianni}) where joint measurability appears and disappears periodically. The reader should contrast this behaviour to the case of semigroup evolution considered in the subsequent section, where incompatibility, once lost, cannot reappear.

\begin{example}\label{ex1}
Let $\bS$ and $\bM$ be two qubit systems, with $\hi=\ki=\C^2$ and $|+\>=(1,0)$, $|-\>=(0,1)$,  and assume that the  evolution $t\mapsto U_t=e^{it H}$ of $\bS+\bM$ is given with respect to the total spin basis  
$
\left\{\ket{++}, \; \frac 1{\sqrt 2}\big(\ket{+-}+\ket{-+}\big),\; \ket{--},\; \frac 1{\sqrt 2}\big(\ket{+-}-\ket{-+}\big)\right\} 
$
where the matrix of $U_t$ is
$$
\begin{pmatrix}
\cos\omega t && -\sin\omega t \\
\sin\omega t&& \cos\omega t
\end{pmatrix}\oplus\id.
$$
In the canonical basis  $\{\ket {++},\ket{+-},\ket{-+},\ket{--}\}$,
$$ 
U_t= \frac 12
\begin{pmatrix}
2\cos\omega t &&-\sqrt 2\sin\omega t &&-\sqrt 2\sin\omega t &&0\\
\sqrt 2\sin\omega t &&\cos\omega t +1&&\cos\omega t -1&&0\\
\sqrt 2\sin\omega t &&\cos\omega t -1&&\cos\omega t +1&&0\\
0&&0&&0&&2
\end{pmatrix}
$$
 and the Hamiltonian is
$
H=-\frac{1}{\sqrt2}\omega\left(P_3\otimes\sigma_2+\sigma_2\otimes P_3\right)
$
where $P_i=\frac 12(\id+\sigma_i)$, $i=1,2,3$, with $\sigma_i$ being the Pauli matrices (also $\sigma_0=\id$).

Let $\sigma=\kb ++$ be the fixed state of $\bM$. 
Using  \eqref{red}
we determine the change $P_i\mapsto \Lambda_t(P_i)$  of the basic qubit projections $P_i$ under the evolution. 
By direct calculation, we get
\begin{eqnarray*}
\Lambda_t(P_1)&=& \frac 12
\begin{pmatrix}
1+\frac1{\sqrt2}\sin(2\omega t)
&&  \frac{1}{2}[\cos(\omega t)+\cos(2\omega t)]\\
\frac{1}{2}[\cos(\omega t)+\cos(2\omega t)] && 1-\frac1{\sqrt2}\left[\sin(\omega t)+\frac12\sin(2\omega t)\right]
\end{pmatrix}
=
\frac 12\sum_{k=0}^3x_k(t)\sigma_k
\\
\big(x_k(t)\big)_{k=0}^3&=& 
\left(1-\frac{\sin(\omega t)}{2\sqrt2}+\frac{\sin(2\omega t)}{4\sqrt2},\;
\cos(\omega t)+\cos(2\omega t),\;0,\;
\frac{\sin(\omega t)}{2\sqrt2}+\frac3{4\sqrt2}\sin(2\omega t)
\right)
\\
\Lambda_t(P_2)&=& \frac 12
\begin{pmatrix}
1&&  -\frac{i}{2}[1+\cos(\omega t)]\\
  \frac{i}{2}[1+\cos(\omega t)]&& 1
\end{pmatrix}
=
\frac 12\sum_{k=0}^3y_k(t)\sigma_k
\\
\big(y_k(t)\big)_{k=0}^3&=&
\left(
1,\;0,\;
\frac12[1+\cos(\omega t)],\;0)
\right) \\
\Lambda_t(P_3)&=& \frac 12
\begin{pmatrix}
1+\cos^2(\omega t)&&  -\frac1{\sqrt2}\left[\sin(\omega t)+\frac12\sin(2\omega t)\right]\\
-\frac1{\sqrt2}\left[\sin(\omega t)+\frac12\sin(2\omega t)\right]&& 
1+\frac12\left[\sin^2(\omega t)-2\cos(\omega t)\right]
\end{pmatrix}.
%=\frac 12\sum_{k=0}^3z_k(t)\sigma_k\\
%\big(z_k(t)\big)_{k=0}^3&=& \left(???\right)\\
\end{eqnarray*}
Without restricting generality, we assume that the shortest period of oscillation is 1, that is, the angular velocity $\omega=2\pi$. Clearly, at times $t=0$ and $t=1$, $\Lambda_t$ is the identity channel and the effects $P_1$, $P_2$, and $P_3$ are not (even pairwise) compatible.
At time $t=\frac12$, $U^L_{1/2}$ is diagonal and all effects $\Lambda_{1/2}(P_1)=\frac12\id$,
$\Lambda_{1/2}(P_2)=\frac12\id$, and $\Lambda_{1/2}(P_3)=\id$ are compatible, 
reflecting (in view of Proposition \ref{JMLambdaprop}) the fact that effects at time $t=1$ cannot be obtained from those at $t=\frac12$ by an application of any quantum channel.
Let us then consider the compatibility of $\Lambda_t(P_1)$ and $\Lambda_t(P_2)$ at an arbitrary time $t$.

Using the compatibility criteria of  \cite[Theorem 1]{Yu_etal_2010} (see also \cite{Guhne_etal_2023})
for the qubit effects,
$\Lambda_t(P_1)$ and $\Lambda_t(P_2)$ are compatible exactly when the function
\begin{eqnarray*}
f(t)&:=&
\sqrt{\frac12\sin^2(\pi t)\left[3+\cos(2\pi t)-2\sqrt2\sin(2\pi t)\right]}+\\
&&+
\sqrt{\frac12\sin^2(\pi t)\left[3+\cos(2\pi t)+2\sqrt2\sin(2\pi t)\right]}
-\cos(2\pi t)-1
\end{eqnarray*}
is not negative. This happens, e.g., in the closed interval $[t_1,t_2]$ where $t_1\approx0.196$ and $t_2\approx0.804$, see Figure \ref{kuva}. 

\begin{figure}[htbp]
\begin{center}
\includegraphics[width=0.7\textwidth]{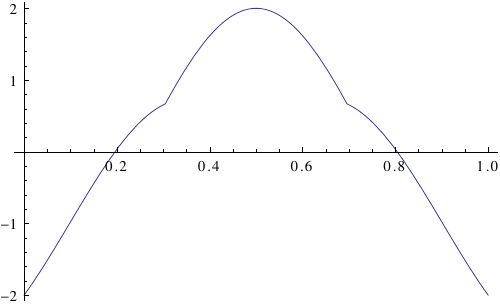} 
\caption{The function $t\mapsto f(t)$ when $t$ runs from 0 to 1.}
\label{kuva}
\end{center}
\end{figure}

\end{example}

\section{Emergent joint measurability under semigroup evolution}

The study of joint measurability in open quantum systems can be seen as a part of the overall effort at understanding the quantum-classical transition due to decoherence. The basic intuition is that decoherence introduces noise to the system, leading to observables eventually becoming ``classical'' in the sense that initially incompatible sets of observables will be jointly measurable when performed at some time $t$ after the start of the evolution. This ``emergence'' of joint measurability can then be taken as an operationally motivated indicator of the quantum-classical transition, which goes well beyond the mere study of decoherence rates or the asymptotic behaviour of the dynamics.

In order to extend the observation in Proposition \ref{JMLambdaprop} to a dynamical setting we need to require that joint measurability at any given time instant implies joint measurability at all subsequent times as well. This can be seen as a form of \emph{Markovianity} in open quantum systems, in the sense of absence of ``revivals'' of incompatibility, which are easily seen to occur in the non-Markovian case \cite{Addis2016}. In Example \ref{ex1}  this is clear, as the system is even periodic, returning to the original measurements after time $2\pi/\omega$.
The general notion of Markovianity has been discussed a lot in the past decades, and is still an active topic; one of the basic definitions is \emph{P-divisibility}, which states that dynamical maps at any two time instants are connected by a positive intermediate transformation, see e.g.\ \cite{Bae}. 
This notion is also the relevant Markovianity from the perspective of measurement incompatibility, but we will here restrict to the semigroup case.

More specifically, we will consider the situation of Theorem \ref{structure}, with a commutative decoherence-free subalgebra. This has the added feature that the limit system is purely classical (which is not true in general, even when the dynamics exhibits decoherence, see \cite{Deschamps2016}). Given any quantum property in the original system, it is then clear that the limit system will not have this property. The nontrivial question, then, is whether the property is lost \emph{already at some finite point in time} when the system is still non-classical, and find bounds for the this critical time --- this is obviously relevant for practical applications of quantum mechanics in noisy systems. While the loss of quantum entanglement has been studied extensively in this sense for decades, much less is known regarding the loss of measurement incompatibility. Nevertheless, it is clear from the above examples that the problem is tractable at least in specific cases, and we will show below that one can indeed obtain explicit bounds in a generic finite-dimensional case, and improved ones in the case of pure decoherence.

The basic starting point for joint measurability considerations is the Markovian property of joint measurability in the semigroup case: 
\begin{proposition}\label{JMSprop}
Let $t\mapsto\Lambda_t$ be a dynamical semigroup, and $\mathsf E,\,\mathsf F$ observables. If $\Lambda_{t_1}(\mathsf E)$ and $\Lambda_{t_1}(\mathsf F)$ are jointly measurable for some $t_1\geq 0$, then $\Lambda_{t_2}(\mathsf E)$ and $\Lambda_{t_2}(\mathsf F)$ are jointly measurable for all $t_2\geq t_1$.
\end{proposition}
\begin{proof}
We have $\Lambda_{t_2}(\mathsf E) = \Lambda_{t_2-t_1}(\Lambda_{t_1}(\mathsf E))$ for all $t_1<t_2$, and any observable $\mathsf E$, and hence any pair of observables jointly measurable at time $t_1$ will also be jointly measurable at time $t_2$ by Proposition \ref{JMLambdaprop}.	
\end{proof}

Of course, initially incompatible observable will not necessarily become jointly measurable at some time, or even asymptotically. However, quite often this does happen, and the task of characterising this in some relevant classes of semigroups is an interesting problem. If $\mathcal H$ is finite-dimensional, the problem can in principle be solved numerically, since joint measurability can be checked at any time point using an appropriate semidefinite program \cite{semi}.

\subsection{Basic examples}
We begin with specific cases where the joint measurability problem is analytically tractable.

\begin{example}
Consider qubit dephasing \cite{Lonigro2022A} of the Lindblad form with a single operator $L_1=\sqrt{\gamma}\sigma_3/2$, where $\gamma>0$ is a a constant describing the strength of dephasing. This is the simplest case of pure decoherence discussed in Section \ref{pure1}, and --- interestingly --- has the special property of arising from an environmental interaction of the form \eqref{red} without any approximation. Indeed, taking $\mathcal K =L^2(\mathbb R)$ with initial state $\sigma=|\psi\rangle\langle \psi|$ such that $|\psi(x)|^2=\tfrac{\gamma}{2\pi}[(x-\Omega)^2+\gamma^2/4]^{-1}$, and the unitary evolution $U_t=e^{-it\sigma_3\otimes Q/2}$, one can easily check that the exact reduced dynamics \eqref{red} gives $\Lambda_t(A) = C_t\star A$, where 
$$
C_t = \begin{pmatrix}1 & e^{it\Omega-t\gamma/2} \\ e^{-it\Omega-t\gamma/2} &1\end{pmatrix},
$$
which has the semigroup form $\Lambda_t= e^{t\bL}$ with the generator
$$\bL(A) = \frac i2\Omega [\sigma_3,A] + \frac 14 \gamma (\sigma_3 A\sigma_3 -A).
$$
In this context, the Hamiltonian $\sigma_3\otimes Q/2$ has been called \emph{shallow-pocket model} by several authors, see e.g.\ \cite{Lonigro2022A} and references therein.

Now consider a binary unbiased qubit observable $\Eo=\{E,\id-E\}$ so that $E=\tfrac 12(\id +{\bf e}\cdot \bm\sigma)$ where ${\bf e}\in \mathbb R^3$, $\|{\bf e}\|\leq 1$, and $\bm\sigma=(\sigma_1,\sigma_2,\sigma_3)$. The evolved observable is therefore determined by the effect $\Lambda_t(A)= \frac 12 (\id + {\bf e}(t)\cdot \bm\sigma)$, where $${\bf e}(t) = \big(e^{-\gamma t}R(t){\bf e}_{0}, \, e_z\big),$$
where the rotation matrix $$R(t)=\begin{pmatrix}\cos \Omega t& \sin \Omega t\\ -\sin\Omega t& \cos\Omega t \end{pmatrix}$$
and the damping factor $e^{-\gamma t}$ acts on the $xy$-components ${\bf e}_{0}:=(e_x,e_y)$ of the vector ${\bf e}$, while the $z$-component remains unchanged.

Since initially unbiased observables remain unbiased throughout this evolution, we can use the basic criterion for joint measurability introduced in \cite{Paul}: indeed, two unbiased qubit observables $\Eo$ and $\Fo$ are jointly measurable at time $t$ if and only if 
$$\|{\bf e}(t)+{\bf f}(t)\|+ \|{\bf e}(t)-{\bf f}(t)\|\leq 2, $$
which reads
$$
\sqrt{e^{-2\gamma t}\|{\bf (e+f)}_{0}\|^2+ (e_z+f_z)^2}+ \sqrt{e^{-2\gamma t}\|{\bf (e-f)}_{0}\|^2+ (e_z-f_z)^2}\leq 2.
$$
In particular, the rotation induced by the Hamiltonian is irrelevant for joint measurability, which only depends on the decoherence rate $\gamma$ and the initial effects. The inequality can be written in the slightly more transparent form
$$
({\bf e}_{0}\cdot {\bf f}_{0})^2 e^{-4\gamma t}-\big[\|{\bf e}_0-{\bf f}_0\|^2+2{\bf e}_{0}\cdot {\bf f}_{0}\,(1-e_zf_z)\big]e^{-2\gamma t} + (1-e_z^2)(1-f_z^2)\geq 0,
$$
which is quadratic in $e^{-2\gamma t}$, so the general solution can easily be written down. It is particularly simple if ${\bf e}_0\cdot {\bf f}_0=0$ (e.g.\ for the projections of $\sigma_1$ and $\sigma_2$), in which case the observables are jointly measurable for 
$$
t\geq -\frac {1}{2\gamma}\ln \frac{(1-e_z^2)(1-f_z^2)}{\|{\bf e}_0-{\bf f}_0\|^2},
$$
provided that $(1-e_z^2)(1-f_z^2)\leq \|{\bf e}_0-{\bf f}_0\|^2$ (otherwise they are jointly measurable from the beginning). 
\end{example}

\begin{example}
Depolarising semigroup with invariant state $\sigma$ reads
$$
\Lambda_{t,\sigma}(A) = e^{-t} A +(1-e^{-t}){\rm tr}[A\sigma]\id.
$$
It is easy to check that the GKLS form is given by the Lindblad operators $L_{ik}=\sqrt{p_k}|\phi_k\rangle \langle i|$ where $\sigma=\sum_k p_k |\phi_k\rangle\langle \phi_k|$ is any convex decomposition of $\sigma$ into pure states, and $\{|i\rangle\}$ is a basis. 

This is one of the few cases where joint measurability results have been derived \cite{Heinosaari2015} using existing fairly sophisticated hidden state models.

Assume that $\mathcal H$ is finite-dimensional with dimension $d$. Let $\{\Eo_1,\ldots, \Eo_n\}$ be $n$ observables with arbitrary finite outcome sets $\Omega_1,\ldots,\Omega_n$. Define for all $x_1\in \Omega_1,\ldots,x_n\in \Omega_n$, a positive operator
$$
\mathsf G\big(\{(x_1,\ldots,x_n)\}\big) := d\binom{d+n-1}{n}^{-1}{\rm tr}_1\big[P_n \Eo_1\big(\{x_1\}\big)\otimes \cdots\otimes\Eo_n\big(\{x_n\}\big)P_n\big],
$$ 
where $P_n$ is the projection to the symmetric subspace of the $n$-fold tensor product of $\mathcal H$, and ${\rm tr}_1$ is the partial trace over all but one tensor factor. It follows that $\mathsf G$ is a joint observable for the observables $\{\Lambda_{t,\id/d}(\Eo_1),\ldots, \Lambda_{t,\id/d}(\Eo_n)\}$ when $t=-\ln ((n+d)/(n(d+1)))$, and hence by Proposition \ref{JMSprop} it follows that they are jointly measurable for
$$
t\geq -\ln \frac{n+d}{n(d+1)}.
$$
Furthermore, it was proved in \cite{Heinosaari2015} that $\{\Lambda_{t,\id/d}(\Eo_1),\ldots, \Lambda_{t,\id/d}(\Eo_n)\}$ is jointly measurable for \emph{any} $n$ if
$$
t\geq -\ln \frac{(3d-1)(d-1)^{d-1}}{(d+1)d^d}.
$$
\end{example}

\subsection{Perturbative approach in the generic case}

\newcommand{\Ec}{\mathcal E_{\mathcal C}}
We now proceed to obtain a rough bound for the critical joint measurability time in a generic system with the structure given in Theorem \ref{structure}. We begin the discussion in a general Hilbert space, but for the concrete result we will have to restrict to a finite-dimensional case.

We first consider joint measurability in the vicinity of a commutative subalgebra $\mathcal C\subseteq \lh$. 
Let $\eh$ denote the set of effects, positive unit bounded operators.
and  let $\Ec= \eh\cap \mathcal C$. 
Obviously, all pairs of $\Ec$--valued observables are jointly measurable, but it is also quite natural to ask how close (pairs of) observables need to be to the algebra $\mathcal C$ to be jointly measurable. To formalise this, we set the following general definition:
\begin{definition}
Let $\mathcal U_0\subseteq \eh\times \eh$ be any set with the property that $[A,B]=0$ for all $(A,B)\in \mathcal U_0$. We call a set $\mathcal U\subseteq \mathcal E(\hi)\times\mathcal E(\hi)$ is a \emph{jointly measurable envelope of $\mathcal U_0$} if $\mathcal U_0\subseteq \mathcal U$, and any pair $(\Eo, \Fo)$ of finite-outcome observables such that $(\Eo(i), \Fo(j))\in \mathcal U$ for all $(i,j)\in \Omega_{\Eo}\times \Omega_{\Fo}$ is jointly measurable.\footnote{For notational simplicity we write here and subsequently, for instance,  $\Eo(i)$ instead of $\Eo(\{i\})$.}
\end{definition}
\begin{remark}
Note that trivially, any $\mathcal U_0$ as in the definition is a jointly measurable envelope of itself, since commutative (discrete) 
observables are jointly measurable. The interesting question is how to find jointly measurable envelopes involving non-commutative pairs. We also remark that the definition could be generalised to include arbitrary many observables (not just pairs) in a straightforward fashion; however, we will restrict to pairs in this paper, as mentioned above.
\end{remark}

We will now apply this concept in the particular case $\mathcal U_0=\Ec\times \Ec$ associated with a commutative algebra $\mathcal C$. In order to find a nontrivial jointly measurable envelope, we assume that there exists a normal conditional expectation $\Gamma:\lh\to \mathcal C$. Since $\Gamma$ is positive and normal, we have, for each observable $\Eo$, the ``projected'' observable $\Gamma(\Eo)$ whose range lies in $\mathcal E_{\mathcal C}$. We now have the pointwise distance $\|\Eo(i)-\Gamma(\Eo)(i)\|$ of the effects from their projections as a measure of closeness,\footnote{In the relevant case where the conditional expectation comes with a faithful invariant state $\sigma$, one can actually regard $\Gamma$ as a Hilbert space projection in the completion of $\lh$ with respect to the scalar product $(A,B)\mapsto {\rm tr}[\sigma A^*B]$, in which case $\Gamma(A)$ is indeed the point in $\mathcal C$ closest to $A$ for any $A\in \lh$.} and the following result holds:

\begin{proposition}\label{JMlope}
Let $\mathcal C\subseteq \lh$ be an Abelian von Neumann algebra with a normal conditional expectation $\Gamma:\lh\to \mathcal C$. For any effect $A$ let $\lambda_0(A)$ denote the bottom of the spectrum of $A$. Then the set
$$
\mathcal U_{\mathcal C,\Gamma} := \{(A,B)\in \mathcal E(\hi)\times \mathcal E(\hi)\mid \|A-\Gamma(A)\|+\|B-\Gamma(B)\|\leq \lambda_0(\Gamma(A))\lambda_0(\Gamma(B))\}
$$
is a jointly measurable envelope of $\mathcal E_{\mathcal C}\times \mathcal E_{\mathcal C}$.
\end{proposition}

\begin{proof}
If $A,\,B\in \mathcal E_{\mathcal C}$, we have $A=\Gamma(A)$ and $B=\Gamma(B)$, and since the right-hand side of the inequality is always nonnegative (as $\Gamma(A),$ $\Gamma(B)\in \mathcal E(\hi)$), it follows that $\mathcal E_{\mathcal C}\times \mathcal E_{\mathcal C}\subseteq{ \mathcal U_{\mathcal C,\Gamma}}$. Now let $\Eo=(\Eo(i))_{i\in \Omega_{\Eo}}$ and $\Fo=(\Fo(j))_{j\in \Omega_{\Fo}}$ be observables with finite outcome sets such that $(\Eo(i), \Fo(j))\in {\mathcal U_{\mathcal C,\Gamma}}$ for each $i,\,j$. For each $i,\,j$, define $$\mathsf G(i,j) := \Gamma(\Eo)(i)\Gamma(\Fo)(j)+ \Eo(i)-\Gamma(\Eo)(i)+ \Fo(j)-\Gamma(\Fo)(j).$$
Now clearly $\sum_j \mathsf G(i,j)=\Eo(i)$ for all $i\in \Omega_{\Eo}$ and $\sum_i \mathsf G(i,j)=\Fo(j)$ for all $j\in \Omega_{\Fo}$, so $\mathsf G(i,j)$ defines an operator measure with margins $\Eo$ and $\Fo$. Since $\Gamma(\Eo)(i)$ and $\Gamma(\Fo)(j)$ commute, we have $\Gamma(\Eo)(i)\Gamma(\Fo)(j) \geq \lambda_0(\Gamma(\Eo)(i))\lambda_0(\Gamma(\Fo)(j))\id$, and hence
\begin{align*}
\mathsf G(i,j)&\geq \lambda_0(\Gamma(\Eo)(i))\lambda_0(\Gamma(\Fo)(j))\id-\|\Eo(i)-\Gamma(\Eo)(i)\|\id-\|\Fo(j)-\Gamma(\Fo)(j)\|\id\geq 0
\end{align*}
because $(\Eo(i), \Fo(j))\in \mathcal U_{\mathcal C,\Gamma}$. Therefore, $\mathsf G$ is a joint observable for the pair $(\Eo,\Fo)$.
\end{proof}

Now consider a semigroup $\Lambda_t$ as in Theorem \ref{structure}, using the notation associated with the decomposition introduced there. For a given effect $A$, we have $\lambda_0(\Gamma(A)) = \inf_{i\in I}{\rm tr}[\sigma_iP_iAP_i]$. Note that since each $\sigma_i$ is faithful, we have ${\rm tr}[\sigma_iP_iAP_i]>0$ for an $i\in I$, provided that $P_iAP_i\neq 0$, but even in the case where $P_iAP_i\neq 0$ for all $i\in I$ it may happen that the infimum is zero. In that case $(A,B)\in \mathcal U_{\mathcal C,\Gamma}$ (for any $B$) only if $A,B\in \mathcal E_{\mathcal C}$. However, in the finite-dimensional case (where $I$ is necessarily finite), $\lambda_0(\Gamma(A))>0$ is a generic property.

The general idea is to show that the effects of a pair of observables $(\Lambda_t(\Eo)$, $\Lambda_t(\Fo))$ enter some jointly measurable envelope of $\Ec\times \Ec$ already at a finite time $t$. In the infinite-dimensional case, we would not expect this as the convergence in part (c) of Theorem \ref{structure} is (only) in the weak-* topology. Therefore, we assume henceforth that $\mathcal H$ is finite-dimensional (which is already nontrivial). In this case, estimates of the stronger form
\begin{equation}\label{Lambdabound}
\|\Lambda_t - \Gamma\circ \Lambda_t\| \leq K e^{-\gamma t}
\end{equation}
where $K,$ $\gamma>0$ are independent of time $t$, have been obtained by various methods; see e.g.\ references cited in \cite{Hanson2020}, in which they were applied to derive bounds for times at which the dynamics becomes entanglement-breaking. Before proceeding to our result, we briefly review this aspect from the point of the view on joint measurability.

Accordingly, we recall that a quantum channel $\Lambda$ is \emph{entanglement-breaking} if $(\Lambda\otimes {\rm Id})(\rho)$ is separable for any bipartite state $\rho$ involving an ancillary system, and that any entanglement-breaking channel in finite-dimensional case is of the form
$$
\Lambda(A) = \sum_{k\in \Omega} {\rm tr}[A\sigma_k] R_k
$$
where each $\sigma_k$ is a state, and $(R_k)_{k\in \Omega}$ is a resolution of identity \cite{horo}. From our point of view, such channels are relevant as they are also \emph{incompatibility-breaking} \cite{Heinosaari2015}, so that, in particular, the pair $(\Lambda(\Eo), \Lambda(\Fo))$ is jointly measurable for \emph{any} pair $(\Eo,\Fo)$ of observables. Indeed, we can obtain any observable of the form $\Lambda(\Eo)$ from $(R_k)_{k\in \Omega}$ by classical postprocessing via Markov kernel $(i,k)\mapsto{\rm tr}[\Eo(i)\sigma_k]$; this property is known to be equivalent to joint measurability in the sense of our definition \cite[Theorem 11.1]{QM}.

In the setting of Theorem \ref{structure}, we immediately observe that the conditional expectation is an entanglement-breaking channel, which agrees with the fact that the algebra $\mathcal C$ is commutative, i.e.\ does not support any incompatible observables. In particular, for any pair of observables $(\Eo,\Fo)$, we can write a joint observable for the pair $(\Gamma(\Eo),\Gamma(\Fo))$ as $\mathsf G(i,j) = \Gamma(\Eo(i))\Gamma(\Fo(j))$ (as the range of $\Gamma$ is commutative), which, using the decomposition in Theorem \ref{structure}, attains the explicit form
$$
\mathsf G(i,j)= \sum_{k\in I} {\rm tr}[\Eo(i)\sigma_k] {\rm tr}[\Fo(j)\sigma_k]P_k,
$$
which is manifestly a postprocessing of the central resolution of identity $(P_k)$ providing a joint observable for all $\mathcal C$--valued observables.

Now, the next question is whether joint measurability could be achieved already at some finite time for a given pair of observables --- this would show that the emergence of classicality, from the point of view of individual pairs of observables, has a practical relevance.

It makes sense to first consider whether $\Lambda_t$ might become entanglement-breaking at some finite time $t$ (as this would obviously imply that it is also incompatibility breaking) for \emph{any} initial pair of observables. The following result is known:
\begin{proposition}[\cite{Hanson2020}] Adopt the assumptions of Theorem \ref{structure} and assume that $\mathcal H$ is $d$--dimensional ($d<\infty$). Then the following hold:
\begin{itemize}
\item[(a)] The semigroup $(\Lambda_t)$ has a finite entanglement-breaking time if and only if the decoherence-free algebra consists exactly of multiples of identity (i.e.\ the index set $I$ has exactly one element). In this case, if \eqref{Lambdabound} holds, then $\Lambda_t$ is entanglement-breaking (and hence also incompatibility breaking) for
$$
t\geq -\frac 1\gamma\ln \frac{\lambda_0(\sigma)}{Kd^{\frac 32}},
$$
where $\sigma$ is the unique faithful invariant state.
\item[(b)] If the algebra $\mathcal C$ is nontrivial ($I$ has more than one element), then there is no finite $t\geq 0$ such that $\Lambda_t$ is entanglement-breaking. 
\end{itemize}
\end{proposition}
This opens up several interesting directions for joint measurability considerations, which we now proceed to initiate. In particular, while part (a) is also a joint measurability result, part (b) is not. We will now provide a simple result showing that incompatibility of a generic pair of finite-outcome observables is indeed broken at some finite time also in this scenario. We call a finite-outcome observable $\Eo$ \emph{regular relative to $\mathcal C$} if $P_k\Eo(i)P_k\neq 0$ for all $k\in I$ and $i\in \Omega_{\Eo}$ (this is the relevant notion of genericity). In this case, we define $\lambda_{\Eo}=\min_{i,k} {\rm tr}[\sigma_k P_k\Eo(i)P_k]>0$. 

\begin{proposition} Adopt again the assumptions of Theorem \ref{structure}. Suppose that $\hi$ is finite-dimensional and \eqref{Lambdabound} holds for $\Lambda_t$. Let $(\Eo, \Fo)$ be a pair of finite-outcome observables which are regular relative to $\mathcal C$. Then $(\Lambda_t(\Eo),\Lambda_t(\Fo))$ is jointly measurable for $$t\geq -\frac{1}{\gamma}\ln \frac {\lambda_{\Eo}\lambda_{\Fo}}{2K}.$$ 
\end{proposition}
\begin{proof}
The assumption \eqref{Lambdabound} implies that $\|\Lambda_t(\Eo(i))-\Gamma(\Lambda_t(\Eo))(i)\|\leq K e^{-\gamma t}\|\Eo(i)\|\leq Ke^{-\gamma t}$ for all $t\geq 0$ and $i\in \Omega_{\Eo}$, and a similar relation holds for $\Fo$. Therefore, if the inequality in the claim holds, we have
\begin{align*}
&\|\Lambda_t(\Eo(i))-\Gamma(\Lambda_t(\Eo))(i)\| + \|\Lambda_t(\Fo(j))-\Gamma(\Lambda_t(\Fo))(j)\|\leq 2 Ke^{-\gamma t}\leq \lambda_{\Eo}\lambda_{\Fo}\\
&\leq \min_{k\in I} {\rm tr}[\sigma_k P_k\Eo(i)P_k] \min_{k\in I}{\rm tr}[\sigma_k P_k\Fo(j)P_k] = \lambda_0(\Gamma(\Eo)(i)) \lambda_0(\Gamma(\Fo)(j))
\end{align*}
for each $i$ and $j$, which shows that each pair $(\Eo(i),\Fo(j))$ belongs to the jointly measurable envelope $\mathcal U_{\mathcal C, \Gamma}$, and the result follows.
\end{proof}

\begin{remark} It is clear that this result would generalise in a straightforward way to joint measurability of $n$--tuples of observables.
\end{remark}

This bound is obviously far from optimal in most specific cases, and can indeed be improved significantly when more information on the structure of the semigroup is available. Since the aim here is mainly expository, we only consider one class of semigroups and observables.

\subsection{The case of pure decoherence}\label{pure}

We now proceed to consider the joint measurability problem in the special case of Section \ref{pure1} where $\mathcal C ={\rm span}\{|n\rangle\langle n| \mid n\in I\}$ in some fixed basis $\{|n\rangle\}\subset\hi$.

For an observable $\Eo$ we then have $\lambda_{\Eo}=\min_{n,j} p_n^{\Eo}(j)$, where we have denoted $p_n^{\Eo}(j)=\langle n|\Eo(j)|n\rangle$, so that $p_n^{\Eo}$ is the probability distribution of $\Eo$ on the basis state $n\in I$.

A tractable and interesting setting for joint measurability considerations is now the following \cite{Kiukas2022}: one of the observables, say, $\Fo$, is $\mathcal C$-valued, so that each effect $\Fo(j)$ is already in the decoherence-free subalgebra (i.e.\ diagonal in the basis $\{|n\rangle\}$). It then follows that $\Fo = \Gamma(\Fo)$. Hence, Proposition \ref{JMlope} applies here: Assuming $\Eo$ and $\Fo$ are regular, their joint measurability is implied by the inequality
\begin{equation}\label{generic2}
\max_i \frac{\|\Eo(i)-\Gamma(\Eo(i))\|}{\min_n p_n^\Eo(i)}\leq \lambda_\Fo  
\end{equation}
However, this result can be improved by using the simplified matrix structure of the observables. First of all, instead of considering the distance $\|\Eo(i)-\Gamma(\Eo(i))\|$ between the individual effects, we use the $\ell^1$--distances of the individual matrix elements, by defining for each $n,\,m$ the \emph{coherence}
$${\rm coh}_{nm}(\Eo):=\sum_i|\<n|\Eo(i)|m\>|.$$
For $n=m$ this just equals 1, while $n\neq m$ produces the relevant quantities. The consideration of the second (diagonal) observable $\Fo$ is more interesting, as joint measurability turns out to depend crucially on how distinct the probability measures $p_n^\Eo$ are. In fact, the following result holds:

\begin{proposition}{\cite[Prop.\ 1]{Kiukas2022}}
If $\Eo$ and $\Fo$ are jointly measurable, then  the following tradeoff holds for all $n,m$:
\begin{equation}\label{tradeoff}
{\rm coh}_{nm}(\Eo)+d^2_{nm}(\Fo)\leq 1,
\end{equation}
where $d^2_{nm}(\Fo):=1-\sum_j\sqrt{p^\Fo_n(j)p^\Fo_m(j)}$ is the squared Hellinger distance between the probability distributions $p^\Fo_n$ and $p^\Fo_m$.
\end{proposition}

This simple but in some sense remarkable inequality can be seen as an ``uncertainty relation'' between the ability of $\Fo$ to distinguish the basis elements, and the ability of $\Eo$ to detect coherences between them. This is, of course, an instance of the fundamental uncertainty principle between complementary observables in quantum mechanics. In fact, the role of complementarity in the present context can be made more explicit by considering mutually unbiased bases (which are also maximally coherent observables). This link to complementarity illustrates the interesting and intricate nature of joint measurability in the neighbourhood of commutative subalgebras.

We can also improve the generic sufficient condition \label{generic2} using the Hellinger distances. In fact:

\begin{proposition}{\cite[Prop.\ 2]{Kiukas2022}}\label{SDP1}
If each matrix $M(i)$ with the matrix elements
$$
M_{nm}(i)=\frac{\<n|\Eo(i)|m\>}{1-d^2_{nm}(\Fo)}
$$
is positive semidefinite, then $\Eo$ and $\Fo$ are jointly measurable (assuming that $d_{nm}(\Fo)\neq 1$ for all $n\neq m$).
\end{proposition}

This also reflects the same complementarity principle, in that the above matrix will typically be positive semidefinite when the off-diagonal is small compared to the diagonal, that is, coherences $\<n|\Eo(i)|m\>$ are small relative to $1-d^2_{nm}(\Fo)$ and the diagonal of $\Eo$ is not too small.

We now proceed to the dynamical setting, given by the Schur channels as noted in Section \ref{pure1}. The relevant joint measurability conditions therefore involve channels of the form $\Lambda(A) = C\star A$ where $C$ is a positive semidefinite matrix with unit diagonal. In a dynamical setting $C$ will depend on time, and we can use the conditions to estimate the critical time at which incompatibility is lost. The point is that initially diagonal observable $\Fo$ will remain invariant under the channel, and we can use the conditions above for $\Eo$ replaced by $\Lambda(\Eo)$ to assess the effect of the dynamics. In this context the following result is known:

\begin{proposition}{\cite[Theorem 2]{Kiukas2022}}\label{SDP}
Let $\Lambda$ be a Schur channel and $\Fo$ a $\mathcal C$-valued (i.e.\ diagonal) observable. Then the pair $(\Lambda(\Eo), \Fo)$ is jointly measurable for \emph{all} observables $\Eo$, if and only if there are positive semidefinite matrices $C^{(j)}$ such that
 $$
 \sum_j C^{(j)}=C, \qquad c_{nn}^{(j)}\equiv p_n^\Fo(j).
 $$
\end{proposition}

This characterisation, which has the form of a semidefinite program (SDP), is numerically much quicker to evaluate than the generic joint measurability SDP. Furthermore, the above two propositions can be used to derive analytical bounds. This was illustrated in \cite{Kiukas2022} in the case of the spin-boson model; we summarise the main results here briefly.

The spin-boson model consists of an $N$--qubit quantum system coupled to an environment consisting of harmonic oscillators in the continuum limit \cite{breuer}; its Hamiltonian is 
$$H = H_S + \sum_k \omega_k b_k^* b_k + \sum_k S_z (g_kb_k^* +\overline g_k b_k),$$
where $b_k$ are the modes (so $[b_k,b_k^*]=1$), $g_k$ the coupling constants, and $H_S = \omega_0 S_z$ the system Hamiltonian with $S_z$ the total spin in the $z$--direction. Taking the initial state $\sigma$ of the environment to be thermal, the resulting subsystem dynamics \ref{red} takes the Schur channel form $\Lambda_t(A) = C[\lambda(t)]\star A$, where $c_{\bf nm}[\lambda]=\lambda^{(|{\bf n}|-|{\bf m}|)^2}$, given in the joint $\sigma_3$--eigenbasis $\{|{\bf m}\rangle\}$ for the qubits, with ${\bf m}=(m_1,\ldots,m_N)\in \{0,1\}^N$ and $|{\bf m}|=\sum_n m_n$. Here $\lambda(t)\in [0,1]$ is given by the decoherence functional of the environment, depending on the spectral density and temperature \cite{breuer}. It follows that the dynamics is a semigroup (only) when $\lambda(t)=e^{-\alpha t}$ for some constant $\alpha>0$, but the Markovian behaviour of joint measurability (Proposition \ref{JMSprop}) already holds when the dynamics is divisible, that is, $\lambda(t)$ is monotone decreasing in $t$.

In order to solve the incompatibility problem, we first eliminate the redundancy related to the decoherence-free subspaces $\mathcal D_k:={\rm span}\big\{|{\bf m}\rangle\,\big|\,|{\bf m}|=k\big\}$, $k=1,\ldots, N$. Clearly, $c_{\bf nm}[\lambda] =1$ when ${\bf n,\,m}$ belong to the same $\mathcal D_k$, and hence these coherences do not decay. Therefore, in particular, incompatible observables supported inside any of these subspaces never become compatible. Consequently, an interesting joint measurability problem of the form of Proposition \ref{SDP} in this setting involves incoherent observables $\Fo$ which do not distinguish between basis elements inside any $\mathcal D_k$, that is, $p^{\Fo}_{\bf n}= p^{\Fo}_{\bf m}$ whenever ${\bf n,\,m}\in \mathcal D_k$ for some $k$; that is, the restriction of $\Fo$ to $\mathcal D_k$ is a trivial observable. It can then be shown that the problem reduces to the consideration of decoherence matrices of the form
$$
C[\lambda] =
\begin{pmatrix} 1 & \lambda & \lambda^4 &\cdots & \cdots &\lambda^{N^2}\\
 \lambda &1 & \lambda & \ddots & & \vdots\\
 \lambda^4 & \lambda & \ddots &\ddots &\ddots & \vdots\\
 \vdots & \ddots & \ddots & \ddots & \lambda & \lambda^4\\
 \vdots & & \ddots & \lambda & 1 &\lambda\\
 \lambda^{N^2} & \cdots & \cdots & \lambda^4 &\lambda & 1\end{pmatrix}.
$$
This matrix has an extra symmetry, namely with respect to reflection about the counter-diagonal. It therefore makes sense to restrict to diagonal observables $\Fo$ having the corresponding \emph{covariance} property, that is, we take $\Fo$ to have outcome set $\{0,\ldots, N\}$, and $p^{\Fo}_{\bf n}(k)=p^{\Fo}_{\bf m}(N-k)$ when $|{\bf n}|=N-|{\bf m}|$. These can be interpreted as unsharp measurements of the label $k$ of the decoherence-free subspace; for further discussion on the role of symmetry in this context, see \cite{Kiukas2022}.

After the above reduction, the problem in Proposition \ref{SDP} can be solved analytically for the simplest case $N=2$. For the general case we restrict to covariant diagonal observables of the form
$$
\Fo_\alpha(k) = \alpha D_k + (1-\alpha)q_k\id,
$$
where $D_k$ is the projection onto the decoherence-free subspace $\mathcal D_k$, $q_k = 2^{-N}{\rm tr}[D_k] = 2^{-N} \binom{N}{k}$, and $\alpha\in [0,1]$. This is natural in the sense that the projections $D_k$ form the spectral resolution of the system Hamiltonian $H_S$; the map $\alpha\mapsto\Fo_\alpha$ then interpolates between the energy observable $H_S$ and its trivialisation, which is the classical coin toss distribution. The problem now has only two parameters, so that joint measurability is characterised by a region in the $(\lambda,\alpha)$-plane. In the case $N=2$ the boundary curve has the analytical form
$$
\alpha = 1-4\lambda^2\big[3+\lambda^4+2\sqrt 2(1-\lambda^2)\big]^{-1}.
$$
For arbitrary $N$, we can use Propositions \ref{SDP1} and \ref{SDP}  to find analytical bounds for the exact boundary curve. In fact, the Hellinger distances between the distributions $p^{\Fo_\alpha}_{\bf n}$ and $p^{\Fo_\alpha}_{\bf m}$ for $|{\bf n}\rangle \in \mathcal D_k$, $|{\bf m}\rangle\in \mathcal D_{k'}$ are given by
$$
d^2_{\bf n,m}(\Fo_\alpha):= \alpha - u_k(\alpha)-u_{k'}(\alpha)
$$
where $u_k(\alpha) =\sqrt{q_k(1-\alpha)(\alpha+q_k(1-\alpha))}-q_k(1-\alpha)$. In particular, taking $k=1,\,k'=0$ gives the inequality
$$
\lambda \leq 1-\alpha + u_0(\alpha)+u_1(\alpha)
$$
as a necessary condition for joint measurability. Furthermore, using the second proposition together with some theory of Toeplitz matrices, one can show \cite{Kiukas2022} that the inequality
$$
\theta_3\left(\frac{\pi}{2}, \lambda\right)\geq \alpha - 2u_0(\alpha),
$$
where $\theta_3(x,\lambda)=1+2\sum_{k=1}^\infty\lambda^{k^2}\cos(2kx)$ is the Jacobi Theta function, implies joint measurability of $(\Lambda_t(\Eo), \Fo)$ for all observables $\Eo$.

These results are obviously highly specific to the case at hand, but they demonstrate that the general methods listed above do produce analytical results for joint measurability problems even in large open quantum systems.

\end{document}